\documentclass[10pt]{article}
\usepackage{amsmath}
\usepackage{amsthm}
\usepackage{graphicx}
\usepackage{epsfig}
\usepackage{amssymb}
\usepackage{multirow}

\begin{document}

\label{firstpage}

\title{A multidimensional maximum bisection problem
}

\author{Zoran Maksimovi\'c \\
Military Academy, generala Pavla Juri\v{s}i\'{c}a \v{S}turma 33, \\11000 Belgrade, Serbia \\
zmaksimovic@beotel.net
}

\maketitle

\begin{abstract}
This work introduces a multidimensional generalization of the maximum bisection problem. A mixed integer linear
programming formulation is proposed with the proof 
of its correctness. The numerical tests, made on the randomly generated graphs, indicates that the multidimensional
generalization is more difficult to solve than the original problem. 
 
\end{abstract}


ACM Computing Classification System (1998): G.1.6

\section{Introduction}

The maximum bisection problem (MBP) is a well known combinatorial optimization problem.
For a weighted graph $G=(V,E)$ with non-negative weights on the edges and where $|V|$ is an even number, 
the maximum bisection problem  consists in finding 
a partition of the set of vertices $V$ in two subsets $S$ and $V\setminus S$, where $|S|=|V\setminus S|$ and the sum of weights of the edges 
between the sets is maximal. The maximum bisection can be applied in different fields such 
as VLSI design \cite{slov}, image processing \cite{shij},  compiler optimization \cite{hand}, etc. 

The maximum bisection problem is NP hard as shown in \cite{garr}. 
The complexity of finding optimal and good solutions of maximum bisection problem 
has given raise to various solution approaches ranging from application algorithms, exact methods to metaheuristics. 

Widely used mathematical formulation with binary variables $x_j$ assigned to each vertex can be presented as 

\begin{flalign*}
\max & \frac{1}{4}\sum_{i,j} w_{ij}(1-x_ix_j)\\
{\rm s.t.\ } &e^T x=0\\
& x^2_j=1,\ \ \ j=1,\ldots ,n\\
\end{flalign*}

where $e\in {R}^n$ is the column vector of all ones, and $^T$ is the transpose operator. It should be noted
that $x_j$ is either $1$ or $-1$ so either $S=\{j|x_j=1\}$ or $S=\{j|x_j=-1\}$. 

This formulation enabled
approximation algorithms based on semidefinite programming. 
Goemans and Williamson approach to maximum bisection in \cite{goemans} was extended by 
Frize and Jerrum in \cite{frieze}
and produced randomized $0.651$ approximation algorithm. In \cite{ye} Ye improved performance ratio to $0.699$
with modification of Frieze and Jerrum approach. 
The approximation ratio was further improved to $0.7016$ by 
Halperin and Zwick in \cite{zwick}, including some triangle inequalities in the
semidefinite programming relaxations.

The main goal of these approaches is the performance guarantee so they are not competitive
with other methods for comparison in computational testing.
In paper \cite{hastad} a proof that there is no polynomial approximation algorithm with performance ratio
greater than $\frac{16}{17}$ is given.  

Beside these approximation algorithms, there are several approaches for its exact solving such as 
linear and semidefinite
branch-and-cut methods \cite{armb},
intersection of semidefinite and  polyhedral relaxations \cite{rendl}.

In \cite{armb} is discussed  the minimum graph bisection problem and branch-and-cut approaches
for finding its solution. The problem definition can be described as follows:

Let $G = (V, E)$ be an undirected graph with $V = \{1, \dots , n\}$ and 
edge set $E \subseteq \{\{i, j\} : i, j \in V, i < j\}$. For given vertex weights $f_v \in {\mathbb{N}} \cup \{0\}, v \in V$, and edge costs $w_{i,j} \in {\mathbb R}$, $\{i, j\} \in E$, a partition of the vertex set $V$ into two disjoint clusters $S$
and $V\setminus S$ with sizes $f(S) =\sum_{i\in S} f_i \le F$ and $f(V \setminus S) \le F$, 
where $F \in {\mathbb {N}}\cap [\frac12 f(V),f(V)]$,
is called a bisection. Finding a bisection such that the total cost of edges in the
cut $\delta(S) := \{\{i, j\} \in E : i \in S \land j \in V \setminus S\}$ is minimal is the minimum bisection
problem (MB). 

If the function $f$ which represents the weight of nodes is equal to one for all nodes
and $F$ is equal to $\frac12 |V|$ and weights on edges $w_{ij}$ takes negative values
this problem becomes the maximum graph bisection problem. In order to apply brunch-and-cut 
approaches authors in \cite{armb} presented an integer linear programming formulation.

It can be assumed without loss of generality that $G$ contains a spanning star rooted at $s$. Indeed, for a selected node $s \in V$ the set of edges can be extended so that $s$ is adjacent to all other nodes in $V$, where the weights $w$ of new edges is equal to zero. 

Let $y_{ij}$ be the binary variables defined as  

$$y_{ij} = \left\{
\begin{array}{rl}
	1, &  {\rm if\ } ij {\rm\ is\ in\ the\ cut}\\
	0, & {\rm otherwise,}
\end{array}
\right.
$$

The mathematial model is formulated as follows:

\begin{flalign*}
\min & \sum\limits_{ij} w_{ij}y_{ij}\\
{\rm s.t.\ } &f_s+\sum\limits_{v\ne s} f_v(1-y_{sv})\le F\\
& \sum\limits_{v\ne s} f_v y_{sv}\le F\\
& \sum\limits_{ij\in C\setminus U} y_{ij}+\sum\limits_{ij\in U} (1-y_{ij})\ge 1, {\rm\ \  cycle\ }C\subseteq E,
{\rm odd\ } U\subseteq C\\
& y\in \{0,1\}^E
\end{flalign*}

Semidefinite programming formulation given in \cite{armb} is very similar to the one
already presented in this paper. 
Separation routines
for valid inequalities to the bisection cut polytope 
is developed and incorporated 
and incorporated into a common branch-and-cut framework for linear and semidefinite relaxations.
On the basis of large sparse instances coming from VLSI design they
showed the good performance of the semidefinite approach versus the mainstream
linear one.

In the paper \cite{rendl} authors presented a method for finding exact solutions of the Max-Cut problem
based on semidefinite formulation. 
Semidefinite relaxation is used and combined with triangle inequalities,
which is  solved with the bundle method. This approach uses Lagrangian duality to get upper 
bounds with reasonable computational effort. 
The expensive part of their bounding procedure is solving the basic semidefinite
programming relaxation of the Max-Cut problem. 
Authors also discussed applicability of their approach on the special case 
of Max-Cut problem where cardinality of partitions is equal i.e.
maximum graph bisection problem.

Another set of approaches, especially for larger scale instances
are metaheuristics. 
From the wide field of applied metaheuristics
let mention some of them such as: memetic search \cite{wu}, variable
neighbeerhood search \cite{ling}, neural networks \cite{fengmin}, 
deterministing anealing \cite{dang}

Memetic search approach presented in \cite{wu} 
integrates a grouping crossover operator and a tabu search optimization procedure.
The proposed crossover operator preserves the largest common vertex groupings
with respect to the parent solutions while controlling the distance between the
offspring solution and its parents. Experimental results indicates that the 
memetic algorithm improves, in
many cases, the best known solutions for MBP.

Variable neighborhood search metaheuristic can obtain high quality solution for max-cut problems.
However, comparing to max-cut problems, max-bisection problems have more complicated feasible 
region via the linear constraint $e^Tx = 0$. It is hard to directly apply the
typical VNS metaheuristic to deal with max-bisection problems. In \cite{ling} Ling et al.
combined the constraint $e^Tx=0$ with the objective function and  
obtained a new optimization problem which is equivalent to the max-bisection problem, 
and then adopted a distinct greedy local search technique to the resulted problem. 
This modified VNS metaheuristic based on the greedy local search technique
is applied to solve max-bisection problems. 
Numerical results indicate that the proposed method is efficient and can obtain high
quality solutions for max-bisection problems.

In \cite{fengmin}, a new Lagrangian net algorithm is proposed to solve max-bisection problems. 
The bisection constraints is relaxed to the objective function by introducing the penalty
function method. A bisection solution is calculated by a discrete Hopfield neural
network (DHNN). The increasing penalty factor can help the DHNN to escape from the
local minimum and to get a satisfying bisection. The convergence analysis of the proposed
algorithm is also presented. Finally, numerical results of large-scale G-set problems show
that the proposed method can find a better optimal solutions.

A deterministic annealing algorithm is proposed for approximating solution of max bisection problem in \cite{dang}. The algorithm is derived from the introduction of a square-root barrier function, where the barrier parameter behaves as temperature in an annealing procedure and decreases from a sufficiently large positive number to 0. The algorithm searches for a better solution in a feasible descent direction, which has a desired property that lower and upper bounds on variables are always satisfied automatically if the step length is a number between 0 and 1. It is proved that the algorithm converges to at least an integral local minimum point of the continuous problem if a local minimum point of the barrier problem is generated for a sequence of descending values of the barrier parameter with zero limit. Numerical results show that the algorithm is much faster than one of the best existing approximation algorithms while they produce more or less the same quality solution.

Any partition of the node set $V$ in two sets defines a set of edges, that we call a {\it cut}, with ends in different partitions. If a graph has weight on edges, than {\it weight of the cut} 
is defined as the sum of weights of edges in the cut. The problem of finding a partition
of the node set where the weight of the cut is maximal is called a Max-Cut problem. From this
definition it follows that there is no restriction on the cardinality of the partitions. 
Maximum graph bisection problem is obtained from Max-Cut problem
if it is required that the partitions have equal cardinality. From the definition 
it follows that the Max-Cut problem is a generalization of the maximum graph bisection problem,
and that maximum graph bisection problem can be solved by introducing 
restrictions about cardinality in Max-Cut problem.

In this paper a multidimensional generalization of maximum bisection problem is introduced, where weights on edges 
instead of numbers are $n$-tuples of positive real numbers. The weight of the cut is the minimum
of sums of  the  coordinates of edge weights. The goal is to find a partition of the set of vertices $V$
in two sets with equal number of vertices and maximal weight of the cut. For $n=1$ we have an ordinary maximum bisection problem. From the fact that maximum bisection problem is $NP$ hard, and that the maximum bisection problem is a special case of the multidimensional maximum bisection problem it follows that multidimensional maximum bisection problem is also $NP$ hard.

The weight of the cut in the multidimensional maximum bisection problem is calculated in two steps: firstly, the coordinates of the weight vectors on the edges of the cut are summed and secondly, the minimum of the sums is determined. This minimum is the weight of the cut. As it can be seen it is more complex than just summing the weights on the edges of the cut, which is the case in the MBP.

Although MMBP is a straightforward generalization of the MBP, most of the existing methods for solving the MBP can hardly be applied to the MMBP.

The semidefinite mathematical  formulation for the MBP cannot be easily transformed to the one for the MMBP. In the MBP semidefinite formulations, the weights of the edges directly  figure in the objective function and they are treated as numbers. In the MMBP,  on each edge a vector of the weights is assigned and we are not interested only in the coordinates, but in the minimum of their sums. This reason makes approximation algorithms based on semidefinite programming presented in previous discussion, notably   
in \cite{frieze}, \cite{goemans},
\cite{zwick}, \cite{karish} and \cite{ye}, inapplicable for solving MMBP. 

Brunch and cut methods based on linear and semidefinite formulations presented in \cite{armb} cannot be applied for several reasons. 
In order to generate the cycles the authors in \cite{armb} introduced the additional edges with the weights equal to zero.

If this method is expanded in the multidimensional variant by introducing the additional edges, having the vectors of the weights of all zeros, a problem will appear: these new edges will be favored in the cut, since the minimum of the sums of the coordinates has to be determined. Also, it is not easy to reformulate the objective function where weight of edges are used.
Since the semidefinite programming formulation is very similar to the one used by the approximation algorithms, the same consideration presented in the previous paragraph can  also be applied in this case.

The method described in \cite{rendl} requires solving the basic semidefinite programming relaxation of the max cut problem
which is case of MMBP cannot be applied  in the case of MMBP, because in MMBP the weights are represented as vectors.

In a proposed memetic search approach presented in \cite{wu} each individual in a population presents a bisection cut. 
If this approach is applied to the MMBP, the calculation of the fitness function could be pretty complicated. Nevertheless, if this approach is sill applied for solving MMBP, the calculation efforts in terms of time will be enormous.

The variable neighborhood search approach proposed in \cite{ling} combines the constraint $e^T x =0$ with the objective function. In the case of MMBP, this approach is not applicable, because the weights are now vectors and the constraint eTx=0 can not be fitted with the objective function.
Also, the greedy local search with a sorting procedure cannot be applied in the case of MMBP since it is unclear which coordinate should be sorted.

Proposed Lagrangian net algorithm in \cite{fengmin} cannot be easily applied on solution
of MMBP. First of all, penalty functions will have to be modified in order to reflect the fact that 
weights are now vectors. Second, the convergence to optimal solution could not be easily translated 
in a such space where weight of edges are vectors instead of numbers.

A deterministic annealing algorithm from [19] can not be easilly applied, since it is not clear what "a feasible descent direction" means in the case of MMBP because the weights are now vectors. Also, the convergence to an integral local minimum also can not be guaranteed in the case of MMBP.

Like many other graph partitioning problems, MBP is applied for solving various practical problems, such as VLSI design \cite{slov}, image processing \cite{shij},  compiler optimization \cite{hand}, social network analysis etc. 
Multidimensional maximum bisection problem appears whenever relation between entities are vectors of numbers instead of single numbers.
Some practical application are:

- For arbitrary pair of workers can be established several aspects of incompatibility. That aspect could be character, knowledge, experience, etc. where the higher level of incompatibility is represented with greater numbers. The problem is to divide the group of workers in two teams with equal size where the greatest part of incompatibility among workers lies between teams. 

- In VLSI design electrical components also have certain aspects that might be considered such as interference, current used, interconnectedness, heat dissipation etc. The problem is to designate electrical components to one of the two boards in such way that, for example, the two warmest components are on the different boards.

\section{Mixed integer solution for the multidimensional maximum bisection problem}

Before the MILP formulation, the formal mathematical formulation of the problem is given. 
Let $G=(V,E)$ be an undirected graph, and $w$ is a function that assigns to 
each edge $e=\{i,j\}$ a $k$-tuple of positive real numbers $(w_{e1},w_{e2},\ldots,w_{ek})$ and $S\subseteq V$. 
The cut $C(S)$ determined by the set $S$ is defined as \\
$$C(S)=\left\{
e\in  E|e\cap S\ne \emptyset \land e \cap (V\setminus S)\ne \emptyset)
\right\}.$$
From the definition, it is obvious that the cuts $C(S)$ and $C(V\setminus S)$ are the same sets.

The weight of the cut  
is defined as 
$$w(C(S))=\min\limits_{1\le l\le k} \sum\limits_{\{i,j\}\in C(S)} w_{ijl}.$$ 
The goal of the multidimensional maximum bisection problem is to 
find a partition of the set of vertices in two sets $S$ and $V\setminus S$ where $|S|=|V\setminus S|$ and where the weight of the cut 
$w(C(S)$ is maximal. 

The multidimensional maximum bisection problem can be illustrated by the example given on the  Figure 1, 
which optimal solution is given with the set $S=\{1,3,5\}$.
The set $S$ generates the cut $C(S)=\{\{1,2\},\{1,4\},\{1,6\},$  $\{2,3\}, \{3,4\},$ $\{3,6\},$ $\{4,5\},$ $\{5,6\}\}$
where the sums over coordinates are $(18,23)$ and the weight of the cut is
$18$.

\begin{figure}[h]
	\centering   
		\includegraphics[width=8cm]{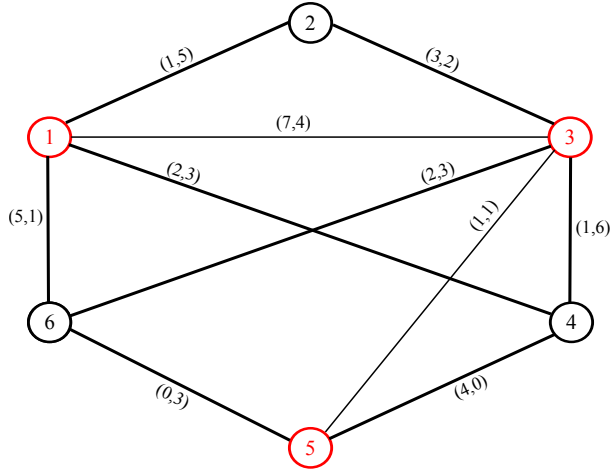}
		\caption{An example of a graph with pairs as weights over the edges}
	\label{fig:example}
\end{figure}

Let $S\subseteq V$, $|V|=n$, $k$ be dimension of weight vector and $w_{el}$ be the $l$-th coordinate of the weight vector for the edge $e$.
\begin{equation*}
\label{eqx}
x_{i}=\left\{
\begin{array}{rl}
	1, &  i\in S\\
	0, & i \notin S,	
\end{array}
\right.\hfill i\in V
\end{equation*}

\begin{equation*}
\label{eqy}
y_e=\left\{
\begin{array}{rl}
	1, & {\rm\ if\ edge\ } e \in C(S)\\
	0, & {\rm\ otherwise,}
\end{array}
\right.\hfill  e\in E
\end{equation*}

The exact solution of the multidimensional maximum bisection problem using mixed integer linear programming 
can be stated as:

\begin{equation}
\label{goal}
\max U
\end{equation}

such that

\begin{equation}
\label{og1}
U \le \sum\limits_{e\in E}w_{el}\cdot y_e,
\ \ \    1\le l\le k
\end{equation}

\begin{equation}
\label{og2}
x_{e_i}+x_{e_j} \ge y_e,
\ \ \    \{e_i,e_j\}=e\in E
\end{equation}

\begin{equation}
\label{og3}
x_{e_i}+x_{e_j}  + y_e\le 2,
\ \ \    \{e_i,e_j\}=e\in E
\end{equation}

\begin{equation}
\label{og4}
\sum\limits_{i=1}^{n} x_i = n/2,
\end{equation}

\begin{equation}
\label{og5}
x_i, y_e \in \{0,1\}, i\in V, e\in E
\end{equation}

\begin{equation}
\label{og6}
U\in[0,+\infty)
\end{equation}

\newtheorem{theorem}{Theorem}

\begin{theorem}
A partition of set of vertices of a given graph $G=(V,E)$ in two sets $S$ and $V\setminus S$ is the solution of the generalized Max-Bisection 
problem if and only if constraints (\ref{og1})-(\ref{og6}) and objective function are satisfied.
\end{theorem}

\begin{proof}
($\Rightarrow$)
Suppose that $S$ is an optimal solution and its corresponding cut is $C(S)$. It will be proved that constraints  ($\ref{og1}$)-($\ref{og6}$) are fulfilled.

Based on the definition of weight of the cut, the constraint ($\ref{og1}$) is true, and based on the goal of the multidimensional maximum bisection problem, (\ref{goal}) also holds.

If $y_e=0$  than ($\ref{og2}$) and ($\ref{og3}$) are obviously true.
If $y_e=1$  than the corresponding edge $e=\{i,j\}$ belongs to the cut, and exactly one vertex incident to the edge $e$ must be in the set $S$, 
so either $x_{e_i}=1$ or $x_{e_i}=1$ and therefore constraints ($\ref{og2}$) and ($\ref{og3}$) holds.

The constraint ($\ref{og4}$) is obviously fulfilled as it is required that the vertex set is partitioned into two set with the equal number of vertices,
and the constraints ($\ref{og5}$) and ($\ref{og5}$) are fulfilled by the definition and the fact that maximum of the cut is to be found. 

($\Leftarrow$) Suppose that objective and constraints are satisfied. The partition of $V$ into two sets $S$ and
$V\setminus S$ is determined by the set 
$S=\{i\in V | x_i=1\}$, where the cut is $C(S)=\{e\in E|y_e=1\}$.

From the constraint ($\ref{og1}$) it follows that 
$$U\le\min\limits_{1\le l\le k} \sum\limits_{\substack{\{i,j\}\in E\\ i\in S, j\notin S}} w_{ijl},$$
meaning that $U\le w(C(S))$ and it follows from the objective function that $U$ is equal to the greatest weight of the cut.

From the constraint ($\ref{og5}$) $y_e$ is either 0 or 1.

If $y_e=1$ then from the constraints ($\ref{og2}$)  and ($\ref{og3}$) it follows that both vertices of the edge $e$ are not in the same set $S$ nor set $V\setminus S$.  

If $y_e=0$ then from the constraints ($\ref{goal}$)-($\ref{og3}$) it follows that both vertices of the edge $e$ must be in the same partitions set (either $S$ or $V\setminus S$). If vertices are in different partitions, than it can be concluded that the weight of the edge $e$ is not included in the weight of the cut, and therefore $U$ is not maximal which contradicts to the supposition that all constraints are fulfilled. From this it follows that vertices of the edge must be in the same partition.  

From the constraint ($\ref{og4}$) it follows that $|S|=n/2=|V\setminus S|$ which means 
that the vertex set is partitioned into two sets with the equal number of vertices.

From the constraints ($\ref{og5}$) it follows that each vector must be in either $S$ or in $V\setminus S$. The same applies for the edges.
\end{proof}

\section{Experimental results}

The experiments were conducted on an Intel Core i3 running on 1.7Ghz with 3GB RAM using CPLEX 12.4, Gurobi 5.6 and total enumeration. In order to run the experiments, a set of 27 random graphs was generated: graphs with 10 vertices and 15, 25 and 40 edges; graphs with 20 vertices and 30, 70 and 150 edges; graphs with 30 vertices and 50, 150 and 400 edges; graphs with 50 vertices and 80, 300 and 1000 edges; graphs with 100 vertices and 150, 500 and 3000 edges; graphs with 300 vertices and 500, 2000, 10000 and 30000 edges; graphs with 500 vertices and 1000, 3000, 10000 and 60000 edges and graphs with 1000 vertices and 1500, 10000, 100000 and 350000 edges.
For each edge of a random graph, a $20$-tuple is generated where each coordinate is a random number in the range $1.000 - 9.999$.  

The experiments were conducted using different vector dimensions: 1, 2, 3, 4, 5, 10, 15 and 20 of the same instances in order to confirm that the increase of the dimension of vectors over the edges  significantly complicate finding of the optimal solution. All tests were run with 7200 seconds time limit. Numerical results for instances where optimal solutions were found is shown in Tables $1 - 3$. In the Tables $4-6$ numerical results are shown for the instances where optimal solutions were not found. 

All tables have common first two columns. In the first column, denoted with {\em instance}, the names of instances are given in the format {\em XXX\_YYY} where {\em XXX} is the number of the vertices and {\em YYY} is the number of the edges. For example, the instance 030\_400 is a graph with $30$ vertices  and $400$ edges. In the second column, denoted with $k$, a vector dimension is given.  

In the Tables $1-3$ in the third column, denoted with {\em opt}, the optimal result is given. The subsequent two columns contain information about total enumeration: time when optimal solution is found ($t$) and total running time ($t_{tot}$). The last four columns contain information about CPLEX and Gurobi time needed for finding optimal solution and running time, denoted in the same manner. 

Third column of the Tables $4-6$ contains the maximum of the solutions found for each method (enumeration, CPLEX and Gurobi).  In the subsequent two columns the solution is given (denoted with sol) and the time needed for finding  that optimal solution (denoted with $t$) using total enumeration. The following four columns contain information about running CPLEX and Gurobi denoted in the same manner. 

As it can be seen in the table $2$ for instance 030\_400 and $k=5,10,15,20$, Gurobi didn't finish their work in 7200 seconds or it run out of memory as well as CPLEX for $k=10,15,20$. 

In the Tables $4-6$ neither of CPLEX, Gurobi and total enumeration complete finding the optimal solutions for the given 7200 seconds for smaller instances and for larger instances because insufficient amount of memory (instances with 1000 vertices).

Obviously, the number of vertices and edges has great influence on the particular solver performance. For example, for the instance 030\_050 with vector dimension 1 both CPLEX and Gurobi completed finding the optimal solution for less than one second, while for the instance 030\_400 it took 775.5 and 193.2 seconds respectively to find the optimal solution. The results, given in the Tables $1-6$, also indicates that the complexity highly increases with the increase of the vector dimension. For example, for the instance 040\_400 where the vector dimension $k$ is equal to 4, it took more than 5000 second for both CPLEX and Gurobi to find the solution.  

Tables $1-6$

\begin{table}[H] \footnotesize
	\caption{Instances with known optimal solutions}
	\centering
\begin{tabular}{ccc|cc|cc|cc}
\hline 
\multirow{2}{*}{{instance}} & \multirow{2}{*}{{$k$}} & \multirow{2}{*}{{opt}} & \multicolumn{2}{c|}{{enumeration}} & \multicolumn{2}{c|}{{CPLEX}} & \multicolumn{2}{c}{{Gurobi}}\tabularnewline
\cline{4-9} 
 &  &  & \multicolumn{1}{c}{$t$~(s)} & $t_{tot}$ (s) & $t$~(s) & $t_{tot}$ (s) & $t$~(s) & $t_{tot}$ (s)\tabularnewline
\hline 
{010\_015} & {1} & {68.708} & {0.001} & {0.001} & {0.051} & {0.155} & {0.030} & {0.030}\tabularnewline
{010\_015} & {2} & {59.971} & {0.001} & {0.001} & {0.046} & {0.062} & {0.030} & {0.030}\tabularnewline
{010\_015} & {3} & {54.324} & {0.001} & {0.001} & {0.066} & {0.088} & {0.040} & {0.040}\tabularnewline
{010\_015} & {4} & {54.324} & {0.001} & {0.001} & {0.091} & {0.121} & {0.040} & {0.040}\tabularnewline
{010\_015} & {5} & {54.324} & {0.001} & {0.001} & {0.085} & {0.212} & {0.050} & {0.050}\tabularnewline
{010\_015} & {10} & {49.011} & {0.001} & {0.001} & {0.263} & {0.269} & {0.050} & {0.050}\tabularnewline
{010\_015} & {15} & {49.011} & {0.001} & {0.001} & {0.134} & {0.141} & {0.050} & {0.050}\tabularnewline
{010\_015} & {20} & {47.347} & {0.001} & {0.001} & {0.165} & {0.239} & {0.060} & {0.060}\tabularnewline
\hline 
{010\_025} & {1} & {82.500} & {0.001} & {0.001} & {0.001} & {0.265} & {0.030} & {0.030}\tabularnewline
{010\_025} & {2} & {82.500} & {0.001} & {0.001} & {0.183} & {0.187} & {0.050} & {0.050}\tabularnewline
{010\_025} & {3} & {82.500} & {0.001} & {0.001} & {0.236} & {0.247} & {0.030} & {0.030}\tabularnewline
{010\_025} & {4} & {82.500} & {0.001} & {0.001} & {0.240} & {0.245} & {0.030} & {0.030}\tabularnewline
{010\_025} & {5} & {82.500} & {0.001} & {0.001} & {0.357} & {0.374} & {0.070} & {0.070}\tabularnewline
{010\_025} & {10} & {79.842} & {0.001} & {0.001} & {0.196} & {0.204} & {0.070} & {0.070}\tabularnewline
{010\_025} & {15} & {79.842} & {0.001} & {0.001} & {0.199} & {0.243} & {0.070} & {0.070}\tabularnewline
{010\_025} & {20} & {79.842} & {0.001} & {0.001} & {0.256} & {0.263} & {0.100} & {0.100}\tabularnewline
\hline 
{010\_040} & {1} & {109.743} & {0.001} & {0.002} & {0.185} & {0.403} & {0.060} & {0.060}\tabularnewline
{010\_040} & {2} & {109.743} & {0.001} & {0.001} & {0.295} & {0.312} & {0.070} & {0.070}\tabularnewline
{010\_040} & {3} & {109.743} & {0.001} & {0.001} & {0.381} & {0.391} & {0.110} & {0.110}\tabularnewline
{010\_040} & {4} & {109.743} & {0.001} & {0.001} & {0.377} & {0.386} & {0.080} & {0.080}\tabularnewline
{010\_040} & {5} & {109.743} & {0.001} & {0.001} & {0.450} & {0.493} & {0.110} & {0.110}\tabularnewline
{010\_040} & {10} & {109.743} & {0.001} & {0.001} & {0.525} & {0.538} & {0.120} & {0.120}\tabularnewline
{010\_040} & {15} & {109.743} & {0.001} & {0.001} & {0.415} & {0.427} & {0.110} & {0.110}\tabularnewline
{010\_040} & {20} & {108.651} & {0.001} & {0.001} & {0.486} & {0.500} & {0.160} & {0.160}\tabularnewline
\hline 
{020\_030} & {1} & {136.696} & {0.016} & {0.078} & {0.068} & {0.227} & {0.030} & {0.030}\tabularnewline
{020\_030} & {2} & {122.664} & {0.031} & {0.078} & {0.105} & {0.109} & {0.040} & {0.040}\tabularnewline
{020\_030} & {3} & {122.664} & {0.031} & {0.078} & {0.047} & {0.194} & {0.040} & {0.040}\tabularnewline
{020\_030} & {4} & {107.134} & {0.015} & {0.093} & {0.214} & {0.220} & {0.050} & {0.050}\tabularnewline
{020\_030} & {5} & {105.441} & {0.046} & {0.109} & {0.267} & {0.276} & {0.090} & {0.090}\tabularnewline
{020\_030} & {10} & {105.441} & {0.062} & {0.124} & {0.290} & {0.298} & {0.110} & {0.110}\tabularnewline
{020\_030} & {15} & {105.441} & {0.062} & {0.140} & {0.196} & {0.203} & {0.090} & {0.090}\tabularnewline
{020\_030} & {20} & {105.441} & {0.475} & {0.171} & {0.249} & {0.258} & {0.090} & {0.090}\tabularnewline
\hline 
{020\_070} & {1} & {250.973} & {0.016} & {0.156} & {0.318} & {0.620} & {0.120} & {0.120}\tabularnewline
{020\_070} & {2} & {250.973} & {0.001} & {0.171} & {0.425} & {0.437} & {0.160} & {0.160}\tabularnewline
{020\_070} & {3} & {246.420} & {0.001} & {0.187} & {0.559} & {0.583} & {0.190} & {0.190}\tabularnewline
{020\_070} & {4} & {246.420} & {0.001} & {0.203} & {0.571} & {0.588} & {0.180} & {0.180}\tabularnewline
{020\_070} & {5} & {246.420} & {0.015} & {0.218} & {0.639} & {0.659} & {0.190} & {0.190}\tabularnewline
{020\_070} & {10} & {244.988} & {0.001} & {0.265} & {0.704} & {0.725} & {0.240} & {0.240}\tabularnewline
{020\_070} & {15} & {244.988} & {0.015} & {0.312} & {0.578} & {0.619} & {0.270} & {0.270}\tabularnewline
{020\_070} & {20} & {244.988} & {0.015} & {0.374} & {0.808} & {0.835} & {0.250} & {0.250}\tabularnewline
\hline 
\end{tabular}
\end{table}

\begin{table}[h] \footnotesize
	\caption{Instances with known optimal solutions}
	\centering
\begin{tabular}{ccc|cc|cc|cc}
\hline 
\multirow{2}{*}{{instance}} & \multirow{2}{*}{{$k$}} & \multirow{2}{*}{{opt}} & \multicolumn{2}{c|}{{enumeration}} & \multicolumn{2}{c|}{{CPLEX}} & \multicolumn{2}{c}{{Gurobi}}\tabularnewline
\cline{4-9} 
 &  &  & $t$~(s) & $t_{tot}$ (s) & $t$~(s) & $t_{tot}$ (s) & $t$~(s) & $t_{tot}$ (s)\tabularnewline
\hline 
{020\_150} & {1} & {502.411} & {0.031} & {0.359} & {1.769} & {2.210} & {0.840} & {0.840}\tabularnewline
{020\_150} & {2} & {476.472} & {0.187} & {0.374} & {1.476} & {1.529} & {1.570} & {1.570}\tabularnewline
{020\_150} & {3} & {466.079} & {0.187} & {0.421} & {1.193} & {2.505} & {2.670} & {1.000}\tabularnewline
{020\_150} & {4} & {466.079} & {0.202} & {0.421} & {2.201} & {2.832} & {2.350} & {2.350}\tabularnewline
{020\_150} & {5} & {455.679} & {0.171} & {0.452} & {4.213} & {5.076} & {2} & {6.760}\tabularnewline
{020\_150} & {10} & {442.486} & {0.202} & {0.561} & {1.912} & {3.387} & {1} & {2.190}\tabularnewline
{020\_150} & {15} & {440.870} & {0.234} & {0.671} & {2.353} & {2.963} & {1} & {5.870}\tabularnewline
{020\_150} & {20} & {440.870} & {0.296} & {0.811} & {3.751} & {3.948} & {1} & {3.070}\tabularnewline
\hline 
{030\_050} & {1} & {224.556} & {47.41} & {117.6} & {0.194} & {0.194} & {0.050} & {0.050}\tabularnewline
{030\_050} & {2} & {224.556} & {49.82} & {123.6} & {0.134} & {0.141} & {0.080} & {0.080}\tabularnewline
{030\_050} & {3} & {221.761} & {54.74} & {136.3} & {0.183} & {0.193} & {0.110} & {0.110}\tabularnewline
{030\_050} & {4} & {221.761} & {54.24} & {145.1} & {0.164} & {0.173} & {0.090} & {0.090}\tabularnewline
{030\_050} & {5} & {214.010} & {158.1} & {151.5} & {0.496} & {0.511} & {0.110} & {0.110}\tabularnewline
{030\_050} & {10} & {207.983} & {86.80} & {198.9} & {0.382} & {0.396} & {0.170} & {0.170}\tabularnewline
{030\_050} & {15} & {199.817} & {92.96} & {233.2} & {0.324} & {0.337} & {0.170} & {0.170}\tabularnewline
{030\_050} & {20} & {199.817} & {111.7} & {279.3} & {0.472} & {0.488} & {0.160} & {0.160}\tabularnewline
\hline 
{030\_150} & {1} & {589.593} & {118.8} & {345.4} & {1.610} & {1.676} & {0.940} & {0.940}\tabularnewline
{030\_150} & {2} & {536.473} & {130.9} & {374.0} & {1.230} & {1.264} & {1.050} & {1.000}\tabularnewline
{030\_150} & {3} & {533.635} & {141.8} & {404.8} & {3.122} & {3.223} & {1.590} & {1.000}\tabularnewline
{030\_150} & {4} & {533.635} & {150.8} & {433.5} & {1.608} & {2.500} & {1.480} & {1.000}\tabularnewline
{030\_150} & {5} & {533.635} & {156.8} & {449.4} & {2.635} & {2.708} & {1.270} & {1.270}\tabularnewline
{030\_150} & {10} & {525.300} & {89.9} & {573.0} & {3.150} & {3.238} & {2} & {3.490}\tabularnewline
{030\_150} & {15} & {519.586} & {106.5} & {936.8} & {3.437} & {4.837} & {1} & {4.470}\tabularnewline
{030\_150} & {20} & {519.586} & {216.7} & {1375} & {3.969} & {4.641} & {1} & {3.120 }\tabularnewline
\hline 
{030\_400} & {1} & {1331.773} & {391.5} & {1327} & {15.08} & {775.5} & {59} & {193.2}\tabularnewline
{030\_400} & {2} & {1230.856} & {539.2} & {1430} & {26.87} & {1928} & {167} & {1249}\tabularnewline
{030\_400} & {3} & {1208.703} & {125.0} & {1416} & {498.5} & {4033} & {2293} & {2889}\tabularnewline
{030\_400} & {4} & {1187.753} & {530.8} & {1520} & {5138} & {5382} & {4389} & {5527}\tabularnewline
{030\_400} & {5} & {1181.987} & {45.5} & {1602} & {5521} & {5563} & {-} & {-}\tabularnewline
{030\_400} & {10} & {1154.353} & {593.1} & {2116} & {-} & {-} & {-} & {-}\tabularnewline
{030\_400} & {15} & {1142.317 } & {519.2} & {2611} & {-} & {-} & {-} & {-}\tabularnewline
{030\_400} & {20} & {1140.635 } & {605.3} & {3098} & {-} & {-} & {-} & {-}\tabularnewline
\hline 
\end{tabular}
\end{table}

\begin{table}[h] \footnotesize
	\caption{Instances with known optimal solutions}
	\centering
\begin{tabular}{ccc|cc|cc|cc}
\hline 
\multirow{2}{*}{{instance}} & \multirow{2}{*}{{$k$}} & \multirow{2}{*}{{opt}} & \multicolumn{2}{c|}{{enumeration}} & \multicolumn{2}{c|}{{CPLEX}} & \multicolumn{2}{c}{{Gurobi}}\tabularnewline
\cline{4-9} 
 &  &  & $t$~(s) & \multicolumn{1}{c}{$t_{tot}$ (s)} & $t$~(s) & $t_{tot}$ (s) & $t$~(s) & $t_{tot}$ (s)\tabularnewline
\hline 
{050\_080} & {1} & {372.069} & {-} & {-} & {0.396} & {0.407} & {0150} & {0150}\tabularnewline
{050\_080} & {2} & {346.271} & {-} & {-} & {0.475} & {0.496} & {0150} & {0150}\tabularnewline
{050\_080} & {3} & {346.271} & {-} & {-} & {0.491} & {0.513} & {0.130} & {0.130}\tabularnewline
{050\_080} & {4} & {346.271} & {-} & {-} & {0.446} & {0.465} & {0.180} & {0.180}\tabularnewline
{050\_080} & {5} & {333.258} & {-} & {-} & {0.697} & {0.720} & {0.200} & {0.200}\tabularnewline
{050\_080} & {10} & {333.258} & {-} & {-} & {0.615} & {0.632} & {0.250} & {0.250}\tabularnewline
{050\_080} & {15} & {332.970} & {-} & {-} & {0.575} & {0.634} & {0.250} & {0.250}\tabularnewline
{050\_080} & {20} & {332.970} & {-} & {-} & {0.882} & {0.911} & {0.250} & {0.250}\tabularnewline
\hline 
{050\_300} & {1} & {1124.331} & {-} & {-} & {7.828} & {23.89} & {15} & {73.58}\tabularnewline
{050\_300} & {2} & {1098.891} & {-} & - & {5.292} & {22.85} & {46} & {60.92}\tabularnewline
{050\_300} & {3} & {1098.891} & {-} & - & {12.82} & {31.82} & {7} & {60.87}\tabularnewline
{050\_300} & {4} & {1094.621} & {-} & - & {51.75} & {117.4} & {21} & {74.23}\tabularnewline
{050\_300} & {5} & {1094.621} & {-} & - & {9.152} & {27.05} & {37} & {89.12}\tabularnewline
{050\_300} & {10} & {1076.105} & {-} & - & {17.71} & {32.11} & {7} & {75.70}\tabularnewline
{050\_300} & {15} & {1062.553} & {-} & - & {29.32} & {295.3} & {45} & {107.30}\tabularnewline
{050\_300} & {20} & {1062.553} & {-} & - & {63.49} & {236.9} & {89} & {163.70}\tabularnewline
\hline 
{100\_150} & {1} & {697.973} & {-} & - & {0.586} & {0.612} & {0.350} & {0.350}\tabularnewline
{100\_150} & {2} & {696.787} & {-} & - & {0.750} & {1.106} & {0.520} & {0.520}\tabularnewline
{100\_150} & {3} & {689.403} & - & - & {1.082} & {1.124} & {0.690} & {0.690}\tabularnewline
{100\_150} & {4} & {689.403} & {-} & - & {0.843} & {1.126} & {0.650} & {0.650}\tabularnewline
{100\_150} & {5} & {689.403} & - & - & {1.425} & {1.425} & {1.090} & {1.090}\tabularnewline
{100\_150} & {10} & {686.703} & - & - & {1.873} & {1.945} & {1.000} & {1.000}\tabularnewline
{100\_150} & {15} & {673.056} & - & - & {1.343} & {1.343} & {1.570} & {1.570}\tabularnewline
{100\_150} & {20} & {655.263} & - & - & {1.678} & {2.410} & {0.930} & {0.930}\tabularnewline
\hline 
{300\_500} & {1} & {2543.862} & - & - & {13.99} & {52.52} & {106} & {127.40}\tabularnewline
{300\_500} & {2} & {2524.215} & - & - & {263.5} & {265.9} & {55} & {279.66}\tabularnewline
{300\_500} & {3} & {2456.085} & - & - & {102.5} & {3455} & {450} & {460.41}\tabularnewline
{300\_500} & {4} & {2456.085} & - & - & {382.3} & {1264} & {151} & {961.72}\tabularnewline
{300\_500} & {5} & {2448.516} & - & - & {116.7} & {1939} & {167} & {752.36}\tabularnewline
{300\_500} & {10} & {2377.530} & - & - & {280.0} & {537.5} & {29} & {287.21}\tabularnewline
{300\_500} & {15} & {2360.732} & - & - & {591.4} & {636.5} & {67} & {654.83}\tabularnewline
{300\_500} & {20} & {2352.357} & - & - & {40.98} & {1766} & {101} & {692.72}\tabularnewline
\hline 
{100\_500} & {1} & {1986.131} & - & - & - & - & {300} & {1495}\tabularnewline
{100\_500} & {2} & {1931.452} & - & - & {292.8} & {4481} & {45} & {1514}\tabularnewline
{100\_500} & {3} & {1901.654} & - & - & {-} & {-} & {1104} & {1998}\tabularnewline
{100\_500} & {4} & {1901.654} & - & - & {-} & {-} & {1435} & {2842}\tabularnewline
{100\_500} & {5} & {1876.418} & - & - & {-} & {-} & {675} & {3902}\tabularnewline
\end{tabular}
\end{table}

\begin{table}[h]\footnotesize
\caption{Instances with unknown optimal solutions}
	\centering
\begin{tabular}{ccc|cc|cc|cc}
\hline 
\multirow{2}{*}{{instance}} & \multirow{2}{*}{{$k$}} & \multirow{2}{*}{{best}} & \multicolumn{2}{c|}{{enumeration}} & \multicolumn{2}{c|}{{CPLEX}} & \multicolumn{2}{c}{{Gurobi}}\tabularnewline
\cline{4-9} 
 &  &  & {sol} & {$t$ (s)} & {sol} & {$t$ (s)} & {sol} & {$t$ (s)}\tabularnewline
\hline 
{050\_1000} & {1} & {3099.083} & {2914.536} & {5125} & {3089.713} & {1531} & {3099.083} & {2404}\tabularnewline
{050\_1000} & {2} & {3016.320} & {2874.113} & {1368} & {3016.320} & {5257} & {3013.999} & {1354}\tabularnewline
{050\_1000} & {3} & {2977.524} & {2861.193} & {2159} & {2977.524} & {650.3} & {2973.205} & {6592}\tabularnewline
{050\_1000} & {4} & {2930.936} & {2826.714} & {2475} & {2930.936} & {3794} & {2928.399} & {2284}\tabularnewline
{050\_1000} & {5} & {2915.139} & {2821.667} & {2586} & {2915.139} & {1636} & {2909.632} & {3419}\tabularnewline
{050\_1000} & {10} & {2890.945} & {2775.508} & {295.9} & {2890.945} & {898.6} & {2827.708} & {1658}\tabularnewline
{050\_1000} & {15} & {2853.382} & {2762.585} & {3324} & {2853.382} & {1401} & {2833.230} & {4129}\tabularnewline
{050\_1000} & {20} & {2840.953} & {2762.585} & {4042} & {2840.953} & {2392} & {2827.708} & {1658}\tabularnewline
\hline 
{100\_500} & {10} & {1852.608} & {1522.286} & {3768} & {1843.131} & {1572} & {1852.608} & {5504}\tabularnewline
{100\_500} & {15} & {1834.379} & {1480.356} & {4748} & {1834.379} & {223.3} & {1834.379} & {392}\tabularnewline
{100\_500} & {20} & {1809.037} & {1480.356} & {5857} & {1809.037} & {607.1} & {1809.037} & {4827}\tabularnewline
\hline 
{100\_3000} & {1} & {8874.360} & {8505.179} & {5089} & {8874.360} & {6230} & {8864.220} & {2720}\tabularnewline
{100\_3000} & {2} & {8718.983} & {8436.693} & {5774} & {8718.983} & {5946} & {8674.793} & {2500}\tabularnewline
{100\_3000} & {3} & {8718.983} & {8386.485} & {5289} & {8656.097} & {7086} & {8718.983} & {3267}\tabularnewline
{100\_3000} & {4} & {8709.104} & {8386.485} & {5692} & {8709.104} & {421.0} & {8699.536} & {30}\tabularnewline
{100\_3000} & {5} & {8698.801} & {8296.34} & {5703} & {8583.056} & {219.3} & {8698.801} & {3170}\tabularnewline
{100\_3000} & {10} & {8709.173} & {8171.752} & {505.9} & {8440.079} & {142.0} & {8709.173} & {4914}\tabularnewline
{100\_3000} & {15} & {8402.444} & {8117.397} & {778.9} & {8402.444} & {267.7} & {8375.000} & {3770}\tabularnewline
{100\_3000} & {20} & {8521.629} & {8048.049} & {4098} & {8464.797} & {266.4} & {8521.629} & {4894}\tabularnewline
\end{tabular}\end{table}

\begin{table}[h]\footnotesize
\caption{Instances with unknown optimal solutions}
	\centering
\begin{tabular}{ccc|cc|cc|cc}
\hline 
\multirow{2}{*}{{instance}} & \multirow{2}{*}{{$k$}} & \multirow{2}{*}{{best}} & \multicolumn{2}{c|}{{enumeration}} & \multicolumn{2}{c|}{{CPLEX}} & \multicolumn{2}{c}{{Gurobi}}\tabularnewline
\cline{4-9} 
 &  &  & {sol} & {$t$ (s)} & {sol} & {$t$ (s)} & {sol} & {$t$ (s)}\tabularnewline
\hline 
{300\_2k} & {1} & {7709.498} & {5744.132} & {6999} & {7690.407} & {6322} & {7709.498} & {6637}\tabularnewline
{300\_2k} & {2} & {7645.349} & {5742.752} & {2774} & {7645.349} & {6487} & {7414.360} & {5032}\tabularnewline
{300\_2k} & {3} & {7509.420} & {5740.743} & {3139} & {7509.420} & {7197} & {7242.234} & {1234}\tabularnewline
{300\_2k} & {4} & {7442.000} & {5739.065} & {4386} & {7442.000} & {6745} & {7402.851} & {1412}\tabularnewline
{300\_2k} & {5} & {7389.467} & {5733.544} & {4731} & {7389.467} & {6871} & {7299.509} & {887}\tabularnewline
{300\_2k} & {10} & {7416.068} & {5733.544} & {7099} & {7416.068} & {7097} & {7168.137} & {2206}\tabularnewline
{300\_2k} & {15} & {7307.361} & {5679.214} & {198.9} & {7307.361} & {6595} & {7291.780} & {4513}\tabularnewline
{300\_2k} & {20} & {7345.519} & {5669.931} & {7136} & {7345.519} & {6059} & {7147.080} & {3663}\tabularnewline
\hline 
{300\_10k} & {1} & {30824.693} & {27486.982} & {2801} & {30824.693} & {0.353} & {29951.782} & {33}\tabularnewline
{300\_10k} & {2} & {29659.475} & {27486.982} & {3249} & {29771.046} & {4288} & {29659.475} & {539}\tabularnewline
{300\_10k} & {3} & {30759.960} & {27472.383} & {3818} & {29819.923} & {6494} & {30759.960} & {336}\tabularnewline
{300\_10k} & {4} & {30989.566} & {27472.383} & {4302} & {29359.477} & {3087} & {30989.566} & {188}\tabularnewline
{300\_10k} & {5} & {30872.004} & {27246.909} & {1004} & {29682.394} & {4164} & {30872.004} & {80}\tabularnewline
{300\_10k} & {10} & {30968.334} & {27246.649} & {931.3} & {29253.722} & {6862} & {30968.334} & {747}\tabularnewline
{300\_10k} & {15} & {30748.223} & {27183.126} & {2194} & {29241.136} & {5092} & {30748.223} & {1299}\tabularnewline
{300\_10k} & {20} & {30923.953} & {27183.126} & {2729} & {29498.737} & {4891} & {30923.953} & {48}\tabularnewline
\hline 
{300\_30k} & {1} & {85934.244} & {83070.887} & {6456} & {85934.244} & {1.225} & {84451.384} & {3}\tabularnewline
{300\_30k} & {2} & {83903.267} & {83028.62} & {6874} & {66065.103} & {5760} & {83903.267} & {7126}\tabularnewline
{300\_30k} & {3} & {84952.867} & {83021.191} & {5183} & {66814.125} & {5575} & {84952.867} & {1573}\tabularnewline
{300\_30k} & {4} & {85143.297} & {83021.191} & {5804} & {66514.511} & {5641} & {85143.297} & {4027}\tabularnewline
{300\_30k} & {5} & {85845.606} & {83028.62} & {6141} & {66007.823} & {6611} & {85845.606} & {3331}\tabularnewline
{300\_30k} & {10} & {85691.919} & {82945.149} & {3427} & {65971.340} & {5955} & {85691.919} & {2977}\tabularnewline
{300\_30k} & {15} & {85894.230} & {82768.105} & {6617} & {66204.979} & {6213} & {85894.230} & {2572}\tabularnewline
{300\_30k} & {20} & {85715.008} & {82735.728} & {6759} & {65252.157} & {7072} & {85715.008} & {3401}\tabularnewline
\hline 
{500\_1k} & {1} & {4744.994} & {2806.965} & {5913} & {4740.014} & {4102} & {4744.994} & {2801}\tabularnewline
{500\_1k} & {2} & {4739.982} & {2806.965} & {7054} & {4739.982} & {6815} & {4739.529} & {6620}\tabularnewline
{500\_1k} & {3} & {4695.510} & {2762.100} & {1590} & {4695.510} & {789.2} & {4692.877} & {4350}\tabularnewline
{500\_1k} & {4} & {4688.866} & {2762.100} & {1239} & {4688.866} & {5575} & {4680.997} & {6234}\tabularnewline
{500\_1k} & {5} & {4687.445} & {2762.100} & {1933} & {4680.664} & {2883} & {4687.445} & {2885}\tabularnewline
{500\_1k} & {10} & {4636.442} & {2762.100} & {2995} & {4636.442} & {2509} & {4623.145} & {5832}\tabularnewline
{500\_1k} & {15} & {4622.783} & {2762.100} & {3794} & {4622.783} & {6846} & {4619.915} & {5316}\tabularnewline
{500\_1k} & {20} & {4602.182} & {2713.144} & {3218} & {4602.182} & {3253} & {4589.822} & {4576}\tabularnewline
\hline 
{500\_3k} & {1} & {11372.166} & {8358.183} & {7124} & {11372.166} & {474.4} & {11289.831} & {2687}\tabularnewline
{500\_3k} & {2} & {11293.525} & {8352.269} & {2402} & {11077.172} & {2835} & {11293.525} & {3663}\tabularnewline
{500\_3k} & {3} & {11257.063} & {8352.269} & {2774} & {11063.300} & {727.4} & {11257.063} & {3975}\tabularnewline
{500\_3k} & {4} & {11267.662} & {8352.269} & {2279} & {10974.964} & {782.7} & {11267.662} & {4074}\tabularnewline
{500\_3k} & {5} & {11214.219} & {8339.367} & {3179} & {10831.476} & {679.7} & {11214.219} & {4509}\tabularnewline
{500\_3k} & {10} & {10972.091} & {8339.367} & {4692} & {10799.479} & {912.4} & {10972.091} & {5591}\tabularnewline
{500\_3k} & {15} & {11159.572} & {8286.828} & {3631} & {10751.672} & {3379} & {11159.572} & {6946}\tabularnewline
{500\_3k} & {20} & {11038.550} & {8286.828} & {5905} & {10812.779} & {4230} & {11038.550} & {6971}\tabularnewline
\end{tabular}
\end{table}

\begin{table}[h]\footnotesize
\caption{Instances with unknown optimal solutions}
	\centering
\begin{tabular}{ccc|cc|cc|cc}
\hline 
\multirow{2}{*}{{instance}} & \multirow{2}{*}{{$k$}} & \multirow{2}{*}{{best}} & \multicolumn{2}{c|}{{enumeration}} & {CPLEX} &  & \multicolumn{2}{c}{{Gurobi}}\tabularnewline
\cline{4-9} 
 &  &  & {sol} & {$t$ (s)} & {sol} & {$t$ (s)} & {sol} & {$t$ (s)}\tabularnewline
\hline 
{500\_10k} & {1} & {31847.822} & {27998.844} & {6100} & {31547.096} & {0.040} & {31847.822} & {1}\tabularnewline
{500\_10k} & {2} & {32453.585} & {27998.844} & {7183} & {31547.096} & {2676} & {32453.585} & {509}\tabularnewline
{500\_10k} & {3} & {32741.433} & {27986.956} & {437.3} & {30709.973} & {4326} & {32741.433} & {272}\tabularnewline
{500\_10k} & {4} & {32134.073} & {27602.152} & {3184} & {30710.467} & {3293} & {32134.073} & {12}\tabularnewline
{500\_10k} & {5} & {32261.153} & {27602.152} & {6980} & {30880.867} & {3453} & {32261.153} & {555}\tabularnewline
{500\_10k} & {10} & {32868.335} & {27602.152} & {6403} & {30610.289} & {5060} & {32868.335} & {416}\tabularnewline
{500\_10k} & {15} & {32579.889} & {27556.169} & {96.57} & {30226.088} & {3517} & {32579.889} & {479}\tabularnewline
{500\_10k} & {20} & {32660.868} & {27556.169} & {189.3} & {30007.443} & {4501} & {32660.868} & {919}\tabularnewline
\hline 
{500\_60k} & {1} & {173626.906} & {165928.512} & {4365} & {173626.906} & {3.894} & {168687.246} & {79}\tabularnewline
{500\_60k} & {2} & {169984.258} & {165762.356} & {1814} & {-} & - & {169984.258} & {2000}\tabularnewline
{500\_60k} & {3} & {171399.820} & {165762.356} & {2057} & {-} & - & {171399.820} & {2937}\tabularnewline
{500\_60k} & {4} & {166948.978} & {165762.356} & {2351} & {-} & - & {166948.978} & {1678}\tabularnewline
{500\_60k} & {5} & {172710.669} & {165762.356} & {3584} & {-} & {-} & {172710.669} & {437}\tabularnewline
{500\_60k} & {10} & {165966.934} & {165615.397} & {3365} & {-} & {-} & {165966.934} & {144}\tabularnewline
{500\_60k} & {15} & {165512.248} & {165067.289} & {3315} & {-} & - & {165512.248} & {14}\tabularnewline
{500\_60k} & {20} & {166051.696} & {165067.289} & {6398} & {-} & {-} & {166051.696} & {4068}\tabularnewline
\hline 
{1k\_1.5k} & {1} & {7830.439} & {4429.553} & {5142} & {7830.439} & {4199} & {7830.439} & {7082}\tabularnewline
{1k\_1.5k} & {2} & {7660.702} & {4406.755} & {3049} & {7660.702} & {500.7} & {7659.381} & {5093}\tabularnewline
{1k\_1.5k} & {3} & {7561.257} & {4331.854} & {2753} & {7561.257} & {394.7} & {7549.120} & {5741}\tabularnewline
{1k\_1.5k} & {4} & {7552.921} & {4331.854} & {3761} & {7552.921} & {6665} & {7551.454} & {2691}\tabularnewline
{1k\_1.5k} & {5} & {7561.257} & {4331.854} & {4126} & {7561.257} & {6775} & {7558.126} & {4247}\tabularnewline
{1k\_1.5k} & {10} & {7499.702} & {4313.156} & {6082} & {7486.617} & {6695} & {7499.702} & {5871}\tabularnewline
{1k\_1.5k} & {15} & {7457.737} & {4212.573} & {2637} & {7457.737} & {6542} & {7448.530} & {5709}\tabularnewline
{1k\_1.5k} & {20} & {7432.879} & {4208.582} & {3862} & {7422.875} & {4892} & {7432.879} & {6725}\tabularnewline
\hline 
{1k\_10k} & {1} & {34497.300} & {27849.420} & {4164} & {34497.300} & {3619} & {33468.985} & {49}\tabularnewline
{1k\_10k} & {2} & {33431.473} & {27453.988} & {3955} & {33431.473} & {2859} & {33308.086} & {6351}\tabularnewline
{1k\_10k} & {3} & {35076.961} & {27449.083} & {2319} & {33788.255} & {2971} & {35076.961} & {81}\tabularnewline
{1k\_10k} & {4} & {34828.276} & {27436.452} & {2905} & {33388.788} & {3102} & {34828.276} & {3937}\tabularnewline
{1k\_10k} & {5} & {35122.442} & {27436.452} & {3200} & {33158.065} & {3165} & {35122.442} & {1306}\tabularnewline
{1k\_10k} & {10} & {34552.295} & {27340.450} & {7001} & {33220.939} & {2952} & {34552.295} & {6163}\tabularnewline
{1k\_10k} & {15} & {34642.263} & {27327.217} & {6969} & {33157.350} & {4653} & {34642.263} & {3900}\tabularnewline
{1k\_10k} & {20} & {34789.795} & {27318.103} & {5733} & {32978.447} & {4106} & {34789.795} & {2139}\tabularnewline
\hline 
{1k\_100k} & {1} & {294083.743} & {272658.632} & {3015} & {294083.743} & {1.156} & {284444.969} & {1254}\tabularnewline
{1k\_100k} & {2} & {277883.845} & {272658.632} & {3860} & {-} & - & {277883.845} & {2339}\tabularnewline
{1k\_100k} & {3} & {277770.801} & {272658.632} & {4468} & {-} & - & {277770.801} & {2931}\tabularnewline
{1k\_100k} & {4} & {287363.683} & {272627.483} & {6531} & {-} & - & {287363.683} & {6850}\tabularnewline
{1k\_100k} & {5} & {293827.978} & {272581.574} & {2464} & {-} & {-} & {293827.978} & {6160}\tabularnewline
{1k\_100k} & {10} & {276657.467} & {272263.665} & {3562} & {-} & {-} & {276657.467} & {2569}\tabularnewline
{1k\_100k} & {15} & {291638.645} & {272127.084} & {3985} & {-} & - & {291638.645} & {5470}\tabularnewline
{1k\_100k} & {20} & {294052.633} & {272114.368} & {5228} & {-} & - & {294052.633} & {1382}\tabularnewline
\hline 
{1k\_350k} & {1} & {986248.932} & {964947.365} & {5146} & {986248.932} & {30.42} & {965489.117} & {4131}\tabularnewline
{1k\_350k} & {2} & {963299.917} & {963270.571} & {6110} & {-} & - & {963299.917} & {1174}\tabularnewline
{1k\_350k} & {3} & {962372.761} & {962136.163} & {3520} & {-} & - & {962372.761} & {3720}\tabularnewline
{1k\_350k} & {4} & {965016.683} & {962136.163} & {3919} & {-} & - & {965016.683} & {1321}\tabularnewline
{1k\_350k} & {5} & {963881.119} & {962136.163} & {4313} & {-} & {-} & {963881.119} & {1551}\tabularnewline
{1k\_350k} & {10} & {962136.163} & {962136.163} & {6521} & {-} & {-} & - & -\tabularnewline
{1k\_350k} & {15} & {962052.538} & {962052.538} & {3701} & {-} & - & - & -\tabularnewline
{1k\_350k} & {20} & {962008.567} & {962008.567} & {2578} & {-} & {-} & - & -\tabularnewline
\hline
\end{tabular}
\end{table}

\section{Conclusions}

This paper has taken into consideration a multidimensional generalization of maximum bisection problem where
weights on the edges are $n$-tuples. A mixed integer linear programming formulation
is introduced with proof of its correctness. Usability of the model is tested on the set of 27 randomly generated graphs with number of vertices ranging from $10$ to $1000$ and number of edges ranging from $15$ to $350000$. The proposed formulation is tested using standard ILP solvers CPLEX and Gurobi, on randomly generated instances. The computational results indicates that the complexity highly increases with the increase of vector dimension especially for the dense graphs.

In future work it may be useful to take into consideration $n$-tuples as weights in several 
related problems, such as Max-Cut, Max $k$-Cut, Max $k$-Vertex Cover, etc.  Other direction could be developing some metaheuristics in cases of large-scale instances which is out of reach for exact methods.


 \label{lastpage}

\begin{thebibliography}{00}


\bibitem{armb}
Armbruster  M., F�genschuh M., Helmberg C., Martin A., 2008. A comparative study of linear and semidefinite
branch-and-cut methods for solving the minimum graph bisection problem. In Proc. Conf. Integer Programming
and Combinatorial Optimization (IPCO), 5035, 112--124.

\bibitem{brun}
Brunetta  L., Conforti M., Rinaldi G., 1997. A branch-and-cut algorithm for the equicut problem. Mathematical
Programming, 78,  243--263.


\bibitem{dang}
Dang C., He L., Hui I.K., 2002.
A deterministic annealing algorithm for approximating a solution of the max-bisection problem
Neural Netw., 15(3), 441--58.


\bibitem{garr}
Garey M. R., Johnson D. S., Stockmeyer L. J., 1976. Some simplified NP-complete graph problems. Theoretical
Computer Science, 1, 237--267.

\bibitem{fengmin}
Fengmin, X.,   Xusheng, M.,  Baili C., 2011.
A new Lagrangian net algorithm for solving max-bisection problems
Journal of Computational and Applied Mathematics 235, 3718--3723


\bibitem{frieze}
Frieze A., Jerrum M., 1997. Improved approximation algorithm for max k-cut and
max-bisection, Algorithmica, 18(1), 67-81.

\bibitem{goemans}
 Goemans, M.X.,  Williamson, D.P., 1995. Improved approximation algorithms for
maximum cut and satisfiability problems using semidefinite programming,
Journal of the Association for Computing Machinery, 42(6):1115-1145.


\bibitem{hage}
Hager W. W., Phan D. T., Zhang H., 2013. An exact algorithm for graph partitioning. Mathematical Programming,
137, 531--556.

\bibitem{zwick}
Halperin, E., Zwick, U., 2002. A unified framework for obtaining improved
approximation algorithms for maximum graph bisection problem, Random
Struct. Algorithms, 20(3), 82-402.


\bibitem{hastad}
 Hastad, J., 1997. Some optimal in approachability results, in: Proceedings of the 29th Annual ACM Symposium on the Theory of Computing, ACM, New York,
1997, pp. 1--10.


\bibitem{hand}
Hendrickson B., Leland R., 1995. An improved spectral graph partitioning algorithm for mapping parallel computations. 
SIAM Journal on Scientific Computing 16(2),  452--469.

\bibitem{karish}
Karish, S.,  Rendl, F.,  Clausen, J., 2000. Solving graph bisection problems with
semidefinite programming, SIAM Journal on Computing 12(3), 177-191.


\bibitem{krish}
Krishnan, K.,  Mitchell, J., 2016. A semidefinite programming based polyhedral cut
and price approach for the Max-Cut problem. Computational Optimization
and Applications, 33(1), 51--71.


\bibitem{ling}
 Ling, A.F.,  Xu, C.X.,  Tang L., 2008.
A modified VNS metaheuristic for max-bisection problems
Journal of Computational and Applied Mathematics 220, 413--421


\bibitem{rendl}
 Rendl, F.,  Rinaldi, G.,  Wiegele, A., 2008. Solving Max-Cut to optimality by intersecting
semidefinite and polyhedral relaxations, Mathematical Programming 121(2),
307-335.


\bibitem{shij}
Shi J., Malik J., 1997.
Normalized Cuts and Image Segmentation. Proceedings of the IEEE
Computer Society Conference on Computer Vision and Pattern Recognition, 731--737.


\bibitem{slov}
Slowik A., Bialko M., 2006.
Partitioning of VLSI Circuits on Subcircuits with Minimal Number 
of Connections Using Evolutionary Algorithm. ICAISC
2006,  Springer-Verlag,  470--478.


\bibitem{wu}
 Wu, Q.,  Hao, J.K., 2013.  Memetic search for the max-bisection
problem, Computers \& Operations Research
40 (1), 166--179.


\bibitem{ye}
Ye, Y., 2001. A 0.699-approximation algorithm for max-bisection, Mathematical
Programming, 90(1), 101--111.
















\end{thebibliography}
\end{document}